\documentclass{lmcs} 
\usepackage{amssymb}

\newcommand{\La}{\langle}
\newcommand{\Ra}{\rangle}

\newcommand{\Let}{\mathsf{let}}
\newcommand{\Spl}{\mathsf{spl}}
\newcommand{\As}{\mathsf{as}}
\newcommand{\In}{\mathsf{in}}

\newcommand{\If}{\mathsf{if}}
\newcommand{\Then}{\mathsf{then}}
\newcommand{\Else}{\mathsf{else}}
\newcommand{\True}{\mathsf{true}}
\newcommand{\False}{\mathsf{false}}
\newcommand{\Bool}{\mathsf{bool}}
\newcommand{\Int}{\mathsf{int}}
\newcommand{\q}{\mathsf{q}}
\newcommand{\ql}{\mathsf{li}}
\renewcommand{\qu}{\mathsf{un}}
\newcommand{\qh}{\mathsf{hi}}

\newcommand{\va}{\mathsf{v}}
\newcommand{\op}{\mathsf{o}}
\newcommand{\e}{\mathsf{e}}

\newcommand{\p}{\mathsf{p}}
\newcommand{\x}{\mathsf{x}}
\newcommand{\y}{\mathsf{y}}
\newcommand{\z}{\mathsf{z}}
\newcommand{\w}{\mathsf{w}}
\newcommand{\n}{\mathsf{n}}
\newcommand{\f}{\mathsf{f}}

\newcommand{\xs}{\mathsf{xs}}

\newcommand{\id}{\mathsf{id}}
\newcommand{\St}{\mathsf{S}}
\newcommand{\T}{\mathsf{T}}

\newcommand{\B}{\mathsf{B}}
\newcommand{\Array}{\mathsf{array}}

\newcommand{\ar}{\mathsf{a}}
\newcommand{\bv}{\mathsf{b}}
\newcommand{\Map}{\mathsf{map}}
\newcommand{\Fib}{\mathsf{fib}}
\newcommand{\Sort}{\mathsf{sort}}
\newcommand{\Swap}{\mathsf{swap}}
\newcommand{\Ins}{\mathsf{ins}}
\newcommand{\vi}{\mathsf{i}}

\newcommand{\vj}{\mathsf{j}}

\newcommand{\Pt}{\mathsf{P}}

\newcommand{\sq}{\varrho}
\newcommand{\sT}{\Upsilon}
\newcommand{\List}{\mathsf{list}}
\newcommand{\Case}{\mathsf{case}}
\newcommand{\Of}{\mathsf{of}}
\newcommand{\zs}{\mathsf{zs}}
\newcommand{\Mapa}{\mathsf{mapa}}
\newcommand{\reg}[2]{$\begin{array}{ll} \underline{#1} \\ #2 \end{array}$}
\newcommand{\regl}[3]{$\begin{array}{ll} \quad #1 \\ \quad #2\\ \overline{#3}\end{array}$}


\begin{document}
\title[Weak-Linear Types]{Weak-Linear Types}
\author[H.~Gramaglia]{H\'ector Luis Gramaglia}
\address{Facultad de Matem\'atica, Astronom\'ia, F\'isica y Computaci\'on - Universidad Nacional de C\'ordoba - CIEM, Centro de Investigaci\'on y Estudios de Matem\'atica}
\email{hector.gramaglia@unc.edu.ar}

\begin{abstract}
%
Computational interpretations of linear logic give rise to efficient methods of static control of memory resources: the data produced by the program are endowed through its type with attributes that determine its life cycle and guarantee safe deallocation. The use of linear types encounters limitations in practice, since linear data, in the traditional sense, do not so often appear in actual programs.
Several alternatives have been proposed in the attempt to relax the condition of linearity, adding coercions to the language to allow linear objects to be temporarily aliased.
In this work, we propose a new alternative, whose virtue is to preserve the simplicity and elegance of the original system.
\end{abstract}

\maketitle

\section{Introduction}
\label{introduccion}
In the formulae-as-types interpretation of Girard's linear logic \cite{girard}, the type of a value is not only a description of its "form", but also, in its computational interpretation, an ability to use it a certain number of times.
This refinement plays a key role in advanced type systems that are developed for a variety of purposes, including static resource management and concurrent systems.
A fundamental advantage of these systems, which Wadler calls ``no discarding'' \cite{wadler}, is the possibility of indicating, with an explicit program directive, when linear data is allocated or deallocated from memory, guaranteeing that the evaluation will not be affected.


Inspired by the works of Wadler \cite{wadler} and Cervesato and Pfenning \cite{cervesato},
the extended Linear Lambda Calculus presented by Walker in \cite{walker} condenses the linear attributes into two features: the decoration with qualifiers\footnote{The two qualifiers warn about the life cycle of a program data: an \textit{unrestricted} data remains in the store supporting multiple aliasing, while a \textit{linear} data is removed from the store after its (only) use.} and the introduction of context splitting, obtaining a conceptually transparent generalization of a classic type system.
Many works address the problem of weakening the notion of linearity for different specific purposes (Wadler \cite{wadler}; Odersky \cite{odersky}; Kobayashi \cite{kobayashi}; Smith, Walker and Morrisett \cite{smith}; Aspinall and Hofmann \cite{aspinall}; Aiken, Foster, Kodumal, and Terauchi \cite{aiken},\cite{foster}).


The main objective of this work is to present a linear applicative language whose type system supports a relaxation of the notion of linearity to allow read-only access to linear data of a base type, but which at the same time preserves the simplicity and elegance of Walker's presentation \cite{walker}.

To achieve our goal, we will introduce a third qualifier, which we will call the "hiding qualifier", and denote $\qh$. The key to this addition is that it does not add a new modality for a program data, but rather
is used internally to promote read-only access by relaxing context splitting, which manages the substructural properties of the system.
It will only be necessary to modify the context split.
As in \cite{walker}, we will use an abstract machine that will make evident the main properties related to memory management.


For a complete description of the history of substructural logics and their applications to Computer Science see \cite{walker} and \cite{dosen}.
Several works use ideas similar to the qualifier $\qh$. We can mention in this line Wadler's \textit{sequential let}  \cite{wadler}, the \textit{usage aspect} given by Aspinally Hofmann in \cite{aspinall}, the \textit{observer annotations} of Oderskyn in  \cite{odersky}, and the  \textit{quasi linear types} of Kobayashi in  \cite{kobayashi}. The qualifier $\qh$ presents similarities with the \textit{use} $\delta$ of \cite{kobayashi}, which constitutes a more general form of weakening of the linearity property.
The distinctive character of our approach is that we retain the main virtue of the formulation given by Walker in \cite{walker}: substructurality is completely captured by the introduction of context splitting, as the only modification to the classical type system.

\section{A linear applicative language}
\label{unlenguajeaplicativo}


Our linear language is built from a qualified signature $\Sigma^\q$, which is defined in Figure 1 from a heterogeneous signature $\Sigma$.

\begin{center}
\begin{tabular}{|c |}
\hline
$
\begin{array}{llll}
\q & ::= & \ql\ |\ \qu\ & \mathsf{qualifiers}\smallskip\\
\varrho & ::= & \ql\ |\ \qu\ |\ \qh& \mathsf{pseudoqualifiers}\smallskip\\
\B & ::= & \Int\ |\ \Bool\ |\  \Array\ |\ ...&  \mathsf{base\ pretypes}\smallskip\\
\tau & ::= & (\sq\ \B,...,\sq\ \B)\rightarrow\q\ \B\ &\mathsf{operators \ types}\smallskip\\
\Sigma^\q & ::= & \{(\op^\tau:\tau)\ :\ (\op:(\B_1,..., \B_n)\rightarrow\B)\in \Sigma,\  & \mathsf{qualified\ signature}\\
&   &  \qquad \qquad \qquad \qquad\qquad\ \tau=(\sq_1\ \B_1,...,\sq_n\ \B_n)\rightarrow\q\ \B \}
\end{array}
$\\
\hline
\end{tabular}
\tiny Figura 1: Qualifiers and qualified signature\end{center}


Qualifying the base types will allow us to obtain different forms of evaluation for our language. Roughly speaking, we have three modalities for a base type $\sq\ \B$
(the qualifier $\qh$ will only be used for base types in the role of input).
The \textit{unrestricted} mode, represented by $\qu\ \B$, indicates that the data can be used an unlimited number of times.
The \textit{linear} mode ($\ql\ \B$) indicates that the data will be used once (without being hidden), and the \textit{hidden} mode ($\qh\ \B$), indicates read-only use of a linear data (it is not deallocated from memory).

The abstract syntax of our language is shown in Figure 2. The abstract phrase $\x$ represents an infinite set of variables.
The phrase $\Spl\ \e\ \As\ \p\ \In\ \e'$ is introduced in \cite{walker} to extract all the components of the tuple counting only one use.

The syntax of types, expressions and contexts is given in Figure 2. As usual, we allow a given variable to appear at most once in a context.

\begin{center}
\begin{tabular}{|c |}
\hline
$
\begin{array}{ll}
\begin{array}{llll}
\Pt & ::= & \B  & \mathsf{pretypes}\\
    &     & \La \T,...,\T\Ra\\
    &     & \T\rightarrow\T\smallskip\\
\T & ::= & \q\ \Pt & \mathsf{types}\smallskip\\
\Upsilon &  ::= &  \T \ | \ \qh\ \B& \mathsf{pseudotypes}\smallskip\\
\Pi    & ::=  & [] \ |\ \Pi,\ \x\!:\!\sT \quad\qquad& \mathsf{type \ context}\smallskip\\
\p & ::= & \La \x,...,\x \Ra&  \mathsf{patterns}\\
\end{array}
&
\begin{array}{rlllll}
\e  &::= & \x & \mathsf{expressions}\\
    &     &  \op^\tau( \e,...,\e)\\
    &    & \q \ \La\e,...,\e\Ra\\
    &    & \e\ \e\\
    &    & \q \ \lambda\x:\T.\e\\
    &    &\Spl\ \e\ \As\ \p\ \In\ \e \\
    &    & \If\ \e\ \Then\  \e\ \Else\ \e \\
    &    &\Let\ \x\equiv \e\ \In\ \e\\
\end{array}
\end{array}
$\\
\hline
\end{tabular}
\tiny Figure 2: Syntax of the linear language
\end{center}


To preserve one of the invariants of linear systems we need to garantee that unrestricted data structures do not hold objects with linear types. To check this,  we define the predicate $\q(\T)$ by the following condition: $\q(\T) $ if and only if $\T=\q'\ \Pt$ and $\q\sqsubseteq\q '$. The relation $\q\sqsubseteq\q'$ is the smallest reflexive and transitive relation that satisfies $\ql\sqsubseteq\qu$. The extension of predicate $\q(\T)$ to type contexts is immediate, as long as we previously extend it to pseudotypes. This is done in a trivial way: $\q(\qh\ \B)=\True$ (refers to the fact that a hidden object has no usage restriction).

The linear type system we present below is based on the system defined by Walker in \cite{walker}. A central device of this system is the \textit{context split} $\Pi_1\circ...\circ \Pi_n=\Pi$, a $(n+1)$-ary relation defined in Figure 3. For simplicity we will define the split for $n=2 $. The reader will have no difficulty in obtaining the definition for the general case.
For convenience, we define the $(0+1)$-ary case as $\qu(\Pi)$\footnote{It is relevant in the rule for $\op^\tau$, when $\op$ is a constant symbol, that is  $\tau=\q\ \B $ (Figure 4).}.

\smallskip

\begin{center}
\begin{tabular}{|c |}
\hline
$
\begin{array}{cc}
\quad\begin{array}{c}
\\
\overline{[] \circ []\ =\ []}\\
\end{array}&
\qquad\begin{array}{c}
\\
\Pi_1\circ\Pi_2=\Pi\qquad(\q\neq\ql)\\
\overline{(\Pi_1,\x: \q\ \Pt) \circ(\Pi_2,\x:  \q\ \Pt)=\Pi,\x:  \q\ \Pt}\\
\end{array}
\\
\\
\quad\begin{array}{c}
\Pi_1\circ\Pi_2=\Pi\\
\overline{(\Pi_1,\x: \ql\ \Pt) \circ\Pi_2\ =\ \Pi,\x:  \ql\ \Pt}\quad\\
\\
\end{array}&
\qquad\begin{array}{c}
\Pi_1\circ\Pi_2=\Pi\\
\overline{\Pi_1 \circ(\Pi_2,\x: \ql\ \Pt)\ =\ \Pi,\x:  \ql\ \Pt}\\
\\
\end{array}
\end{array}
$\\
\hline
\end{tabular}
\tiny Figure 3: Context split
\end{center}


But the context split, which is suitable for the typing of terms, is not suitable for the typing of expressions in general. By typing these, we must generate the possibility of a hidden use of a data as input of a basic operation.
For this we define the \textit{context pseudosplit}, for which we will use the $\sqcup$ operator. Its definition coincides with the definition of the context split, except in the case of a linear base type (that is, a type of the form $\ql\ \B$). In this case the occurrence of $\x:\ql\ \B$ in the $i$-th context is preceded by occurrences of $\x$ as a hidden object. In the following rule $j$ takes the values $1,...,n$.

\begin{center}
$
\begin{array}{c}
\Pi_1\sqcup...\sqcup\Pi_n=\Pi\\
\overline{(\Pi_1,\x: \qh\ \B) \sqcup... \sqcup(\Pi_{j-1},\x: \qh\ \B)\sqcup(\Pi_{j},\x: \ql\ \B)\sqcup\Pi_{j+1}\sqcup...=\Pi,\x:  \ql\ \B}\\
\end{array}
$
\end{center}

To express the fact that an argument of a basic operator can be an expression of type $\q\ \B_i$, or a variable of pseudotype $\qh\ \B_j$, we introduce the \textit{pseudotyping relation} $\ \Pi\Vdash\e:\T$ as the extension of the relation $\Pi\vdash\e:\T$ with the following rule: $\qu(\Pi_1,\Pi_2)$ implies $\Pi_1,\x:\qh\ \B,\Pi_2\Vdash\x:\qh \ \B$.

The rules of the linear type system are given in Figure 4.

\begin{center}
\begin{tabular}{|c |}
\hline
$
\begin{array}{ll}
\begin{array}{l}
\\
\quad\qu(\Pi_1)\\
\quad\qu(\Pi_2)\\
\overline{\Pi_1, \x :\q\ \B,\Pi_2\vdash \x :\q\ \B}\\
\end{array}
&
\begin{array}{l}
\\
\quad\tau=\La\varrho_1\ \B_1,...,\varrho_n\  \B_n\Ra\rightarrow\q\ \B \\
\quad\Pi_i \Vdash \e_i:\sq_i\ \B_i\\
\overline{\Pi_1\circ...\circ\Pi_n\vdash \op^\tau(\e_1,...,\e_n):\q\ \B}\\
\end{array}
\bigskip\\
\begin{array}{l}
\quad\Pi_1\vdash \e:\q\ \La\T_1,...,\T_n\Ra\\
\quad\Pi_2,\x_1:\T_1,...,\x_n:\T_n\vdash \e':\T'\\
\overline{\Pi_1\sqcup\Pi_2\vdash\Spl\ \e\ \As\ \La\x_1,...,\x_n\Ra\ \In\ \e':\T'}\qquad\\
\end{array}
&
\begin{array}{l}
\quad\Pi_1\vdash \e_1:\T_1\quad \q(\T_1)\ ...\\
\underline{\quad\Pi_n\vdash \e_n:\T_n\quad \q(\T_n)\qquad \quad\quad}\\
\Pi_1\sqcup...\sqcup\Pi_1\vdash \\
\quad\q\ \La\e_1,...,\e_n\Ra:\q\ \La\T_1,,...,\T_n\Ra\\
\end{array}
\bigskip\\
\begin{array}{l}
\quad\q(\Pi)\qquad \\
\quad\Pi,\x:\T\vdash \e:\T'\\
\overline{\Pi\vdash \q\ (\lambda \x:\T.\ \e) :\q\ (\T\rightarrow\T')}\\
\end{array}
&
\begin{array}{l}
\quad\Pi_1\vdash \e:\q\ \Bool\\
\quad\Pi_2\vdash \e_i:\T\\
 \overline{\Pi_1\sqcup\Pi_2\vdash \If\ \e\ \Then\ \e_1\ \Else\ \e_2:\T}
\end{array}
\bigskip\\
\begin{array}{l}
\quad\Pi_1\vdash \e:\T\\
\quad\Pi_2,\x:\T\vdash \e':\T'\qquad\qquad\\
\overline{\Pi_1\sqcup\Pi_2\vdash
\Let\ \x\equiv \e\ \In\ \e':\T'}\\
\\
\end{array}
&
\begin{array}{l}
\quad\Pi_0\vdash \e_0:\T\\
\quad\Pi_1\vdash \e_1:\T\rightarrow\T'\\
\overline{\Pi_0\sqcup\Pi_1\vdash \e_1\ \e_0:\T'}\\
\\
\end{array}
\end{array}
$\\
\hline
\end{tabular}\\
\tiny Figure 4: $\Pi\vdash\e:\T$.
\end{center}

Prevalues, values and store are defined in Figure 5. By $\op$ (without arguments) we denote the constants of $\Sigma$, that is, the function symbols of arity 0.

\begin{center}
\begin{tabular}{|c |}
\hline
$
\begin{array}{ll}
\begin{array}{llll}
\\
\w & ::= & \op \ |\   \La\x,...,\x\Ra\ |\ \lambda\x\!:\!\T.\e & \mathsf{prevalues}\smallskip\\
\va & ::= &  \q\ \w \quad & \mathsf{values}\\
\\
\end{array}
\qquad\qquad&
\begin{array}{llll}
\St & ::= & [] &\mathsf{Stores}\\
& &  \St,\x=\va\\
\end{array}
\end{array}
$\\
\hline
\end{tabular}\\
\tiny Figure 5: Prevalues, values and store.
\end{center}

The rules for typing the store are given in Figure 6. We must add to the rules given in \cite{walker} the rule $\mathsf{(shi)}$ for hidden data. The relation $\vdash (\St,\e)$, which indicates that the pair formed by the store and the program are suitable to be evaluated, is also defined in Figure 6.

\medskip

\begin{center}
\begin{tabular}{|c |}
\hline
$
\qquad\begin{array}{ll}
\mathsf{(sem)}
\begin{array}{c}
\underline{\quad\quad\quad\quad}\\
\ \vdash []:[]\\
\end{array}
\qquad\ \
&
\mathsf{(shi)}\begin{array}{c}
\\
\underline{\vdash \St: \Pi\qquad (\va:\ql\ \B)\in\Sigma^\q}\qquad\\
\vdash (\St,\x=\va):(\Pi,\x:\qh\ \B)\\
\\
\end{array}
\medskip\\
\mathsf{(sva)}\begin{array}{c}
\underline{\vdash \St: \Pi_1\circ\Pi_2\qquad \Pi_1\vdash \va:\T}\\
\vdash (\St,\x=\va):(\Pi_2,\x:\T)\\
\end{array}
\quad&
\qquad
\begin{array}{c}
\underline{\vdash\St: \Pi\qquad \Pi\vdash \e:\T}\\
\vdash (\St,\e)\\
\end{array}
\medskip\\
\end{array}
$\\
\hline
\end{tabular}\\
\tiny Figure 6:  Relations $\vdash\St:\Pi$ and $\vdash(\St,\e)$
\end{center}

\subsection{Small-step semantic}
\label{semanticassmallstep}

Different ways of qualifying the operators $\op^\tau$, with $\tau=( \sq_1\ \B_1,...,\sq_n\ \B_n)\rightarrow\q\ \B$,
will give rise to various forms of evaluation, which will differ in the form of manage memory resources. Note that there is an overspecification of qualifiers: in a traditional linear system, the information given by $\sq_1,...,\sq_n$ would not be necessary, since the qualifier of the stored value determines its lifecycle. In our system, the pseudoqualifiers $\sq_1,...,\sq_n$ enable the possibility that a linear value is not discarded when it corresponds to an input of type $\qh\ \B$.

\textit{Evaluation context} $\e[]$ and the \textit{context rule} (cct) are defined in Figure 7.

\begin{center}
\begin{tabular}{|c |}
\hline
$
\begin{array}{l}
\begin{array}{lllllll}
(\St_0,\e_0) \rightarrow_\beta (\St_1, \e_1)\\
\overline{(\St_0, \e[\e_0]) \rightarrow (\St_1,\e[\e_1])}\\
\end{array}
\qquad\qquad\begin{array}{lllllll}
\e[] &::=  & []  & (\textsf{context hole)}\\
 &    & \op^\tau(\x,...,\e[],...,\e)\\
&     & \q\  \La \x,...,\e[],..,\e\Ra\\
     &     & \x\ \e[]\\
     &    &\If\ \e[]\ \Then\ \e\ \Else \ \e\qquad\\
     &     &\Spl\ \e[]\ \As\ \p\  \In\ \e\\
     &     &\Let\ \x\equiv\e[]\ \In\ \e\\
\end{array}\\
\end{array}
$\\
\hline
\end{tabular}
\tiny Figure 7: Contexts evaluation and contexts rule.
\end{center}

To define the small-step semantics we will use the context-based semantics used in \cite{walker}. Its distinctive characteristic is the explicit management of the store $\St$, for which we assume that no variables are repeated, and that when extending it a new variable is used, supplied by $new\ \St$.

A sequence of the form $\x_1\mapsto\y_1,...,\x_n\mapsto\y_n$ will denote a substitution in the usual way: $\x_1\mapsto\y_1$ denotes the identity map, modified in the variable $ \x_1$, where it takes the value $\y_1$.
Furthermore, if $\delta$ is a substitution, then the modified substitution $(\delta,\x\mapsto \y)$ is defined by the conditions $(\delta,\x\mapsto \y)\ \x = \ y$, and $ (\delta,\x\mapsto \y)\ \z = \delta\z$ if $\z\neq\x$. We often will write $\La\x_1,...,\x_n\Ra\mapsto\La\x'_1,...,\x'_n\Ra$ instead of $\x_1\mapsto\x_1,...,\x_n\mapsto\x_n$.

To represent memory deallocation we will use the operator $\sim_{\sq_1,...,\sq_n}$, defined by the following conditions:
\medskip

$
\begin{array}{rlll}
(\St,\x=\va,\St')\sim_\ql\x & = &\St,\St'\\
\St\sim_\sq\x & = &\St\qquad(\sq\neq\ql)\\
\St\sim_{[]}\mathsf{[]} & = &\St\\
\St\sim_{\sq,\sq\mathsf{s}}\x,\mathsf{xs} & = &(\St\sim_\sq\x)\sim_{\sq\mathsf{s}}\mathsf{xs}\\
\end{array}
$

\medskip

Terminal configurations will be pairs of the form $(\St,\x)$. The rules of small-step semantics are given in Figure 8.
We can observe in the rule $\mathsf{(eop)}$ that the reading of the qualifier of each input is done from the specification $\tau$, and not from the store. In this way, hidden data is prevented from being deleted.

\begin{center}
\begin{tabular}{|c |}
\hline
$
\begin{array}{lll}
\mathsf{(eva)} & (\St,\va) \rightarrow_\beta(\St,\x=\va,\x)& (\x=new\ \St)\smallskip\\
\mathsf{(eop)} & (\St,\op^\tau(\x_1,...,\x_n)) \rightarrow_\beta&(\St\x_i=\sq'_i\ \w_i,\ \x=new\ \St,    \\
&\qquad\qquad(\St\sim_{\sq_1,...,\sq_n}\x_1,...,\x_n,\x=\op(\w_1,...,\w_n),\x)\quad&\ \tau =( \sq_1\ \B_1,...\ )\rightarrow\q\ \B)\smallskip\\
\mathsf{(eif)} & (\St, \If\ \x\ \Then\ \e_0\ \Else\ \e_1)\rightarrow_\beta(\St\sim_\q\x,\e_0) & (\St\x=\q\ \True)\smallskip\\
\mathsf{(eif)} & (\St, \If\ \x\ \Then\ \e_0\ \Else\ \e_1)\rightarrow_\beta(\St\sim_\q\x,\e_1) & (\St\x=\q\ \False)\smallskip\\
\mathsf{(esp)} & (\St, \Spl\ \x\ \As\ \p\ \In\  \e) \rightarrow_\beta (\St\sim_\q\x, [\p\mapsto\p']\e) &(\St\x=\q\ \p')\smallskip\\
\mathsf{(efu)} & (\St, \f\ \x') \rightarrow_\beta (\St\sim_\q\f, [\x\mapsto\x']\e) & (\St \f=\q\ \lambda\x:\T.\e)\smallskip\\
\mathsf{(ele)} & (\St, \Let\ \x\equiv \y\ \In\  \e) \rightarrow_\beta
(\St,[\x\mapsto\y]\e), & \smallskip \\
\end{array}
$\\
\hline
\end{tabular}
\tiny Figure 8: Evaluation Rules
\end{center}

\section{Preservation and progress}
\label{preservacionyprogreso}


The fact of using two different splits is justified by the need to guarantee that the small-step semantics preserve the typing. Indeed, suppose we use pseudosplit of contexts in the rule for $\op^\tau$. Then we could prove $\vdash (\x=\ql\ 3,\y=\ql\ 1;\x+^\tau\x*^{\sigma}\y)$, with $\tau = \La\qh\ \Int,\ql\ \Int\Ra\rightarrow\ql\ \Int$ and $\sigma= \La\ql\ \Int,\ql\ \Int\Ra\rightarrow\ql\ \Int$.
But in this case there would be no preservation (nor progress), since $\vdash (\z=\ql\ 2,;\x+^\tau\z)$ could not be proven. On the other hand, expressions like $\La \x,\x*^{\sigma}\y \Ra$ show the need for $\qh\ \B$ to be a pseudotype (type that only occurs as input to an operator).

We now give some notation about contexts that we will use in this section. Given a statement $\mathcal{S}(\Pi)$ that refers to a generic context $\Pi$, we use $\mathcal{S}(\Pi_1\circ...\circ\Pi_n)$ to denote the existence of a context $\Pi$ such that $\Pi = \Pi_1\circ...\circ\Pi_n$ and $\mathcal{S}(\Pi)$.
We will say that the context $\Pi_2$ is \textit{complementary to} $\Pi_1$ if there exists $\Pi $ such that $\Pi=\Pi_1\circ\Pi_2$. Note that for each context $\Pi$ there exists only one context $\Pi^\qu$ that satisfies $\Pi=\Pi\circ\Pi^\qu=\Pi^\qu\circ\Pi$. This context is denoted by $\Pi^\qu$  since it is the subcontext of $\Pi$ formed by the variables whose types $\T$ satisfy $\qu(\T)$.

\begin{lem}[Well-typing] If  $\vdash \St:\Pi$, then:
\begin{enumerate}
    \item $\Pi\x=\q\ \B$ implies $\St\x=\q\ \bv$, with $(\bv:\B)\in\Sigma$.
    \item  $\Pi\x=\qh\ \B$ implies $\St\x=\ql\ \bv$, with $(\bv:\B)\in\Sigma$.
    \item $\Pi\x=\q\ \La\T_1,...,\T_n\Ra$ implies $\St\x=\q\ \La\x_1,...,\x_n\Ra$, for some variables $\x_1,...,\x_n$.
    \item $\Pi\x=\q\ (\T\rightarrow\T')$ implies $\St\x=\q\ (\lambda\y:\T.\e)$, for some $\y,\e$.
\end{enumerate}
\end{lem}

\begin{proof}
Suponga que $\Pi=\Delta,\x\!:\!\sT,\Delta'$. Hacemos induccion sobre $\Delta'$. Claramente la dificultad esta en el caso base, el paso inductivo es trivial. 

\bigskip

\noindent Caso $\sT=\qh\ \B$: Directo de \reg{\vdash \St: \Pi\quad (\va:\ql\ \B)\in\Sigma^\q\ \ }
{\vdash (\St,\x=\va):(\Pi,\x:\qh\ \B)}

\bigskip

\noindent  Caso $\sT=\q\ \B$: Directo de
\reg{\vdash \St: \Pi_1\circ\Pi_2\quad \Pi_1\vdash \va:\q\ \B}
{\vdash (\St,\x=\va):(\Pi_2,\x:\q\ \B)}.
Here $\Pi_1=\Delta^\qu$ and $\Pi_2=\Delta$.

\bigskip
\noindent  Caso $\sT=\q\ \La\T_1,...,\T_n\Ra$: Directo de
\reg{\vdash \St: \Pi_1\circ\Pi_2\quad \Pi_1\vdash \va:\q\ \La\T_1,...,\T_n\Ra}
{\vdash (\St,\x=\va):(\Pi_2,\x:\q\ \La\T_1,...,\T_n\Ra)}. 

Here $\Pi_1\vdash \va:\q\ \La\T_1,...,\T_n\Ra$ implies $\va=\q\ \La\x_1,...,\x_n\Ra$.

\bigskip
\noindent  Caso $\sT=\q\ (\T_1\rightarrow\T_2)$: Directo de
\reg{\vdash \St: \Pi_1\circ\Pi_2\quad \Pi_1\vdash \va:\q\ (\T_1\rightarrow\T_2)}
{\vdash (\St,\x=\va):(\Pi_2,\x:\q\ (\T_1\rightarrow\T_2))}. 

Here $\Pi_1\vdash \va:\q\ (\T_1\rightarrow\T_2)$ implies $\va=\q\ \lambda\x\!:\!\!\T.\e$.
\end{proof}

We call $\beta$-\textit{node} the expressions that occur in the rules of $\rightarrow_\beta$ (Figure 8), that is, expressions of one of the following forms: $\va$, $\op^\tau(\x_1,...,\x_n)$, $\If\ \x\ \Then\ \e_0\ \Else\ \e_1$, $\Spl\ \x\ \As\ \p\ \In\  \e$, $\x\ \x'$ and  $\Let\ \x\equiv \y\ \In\  \e$.

\begin{cor}[$\beta$-progress]If $\vdash(\St,\e)$ and $\e$ is a $\beta$-node, then there exist  $\St',\e'$ such that  $(\St,\e)\rightarrow(\St',\e')$. \qed 
\end{cor}

\medskip
\begin{lem}[Split] If $n>0$ then:
\begin{enumerate}
 \item  If $\vdash \St:\Pi_1\circ...\circ\Pi_n$,  then there exist $\St_1,...,\St_n$ such that    $\vdash \St_j:\Pi_j$, for all $j=1,...,n$. Moreover, $\St_1,...,\St_n$ are all subsequences of $\St$ that, taken in pairs, do not have variables with linear values in common.
 \item  If $\vdash \St:\Pi_1\sqcup...\sqcup\Pi_n$,  then there exist $\St_1,...\St_n$ such that    $\vdash \St_j:\Pi_j$, for all $j=1,...,n$. Moreover, $\St_1,...,\St_n$ are all subsequences of $\St$ that, taken in pairs, they do not have variables with linear values in common, except those of base type.
\end{enumerate}
\end{lem}

\begin{proof} (1) The proof is by structural induction on  $\vdash \St:\Pi$, where $\Pi=\Pi_1\circ...\circ\Pi_n$. The delicate part of the argument is the case $\St=(\St_0,\x=\va)$, $\Pi=\Delta_2,\x:\ql\ \Pt$, obtained from:

 \smallskip

$
\begin{array}{ll}
\underline{\vdash\St_0:\Delta_1\circ\Delta_2\quad \Delta_1\vdash \va:\ql\ \Pt}\\
\vdash \St_0,\x=\va:\Delta_2,\x:\ql\ \Pt
\end{array}
$.

\smallskip

\noindent Let $j$ such that $\Pi_j=\Pi^0_j,\x:\ql\ \Pt$. Then $\Delta_2=\Pi_1\circ...\circ\Pi^0_j\circ ... \circ\Pi_n$, and hence there exists $\Pi'_j$ such that $\Pi'_j=\Delta_1\circ\Pi^0_j$ and $\Delta_1\circ\Delta_2=\Pi_1\circ...\circ\Pi'_j\circ ... \circ\Pi_n$, because the variables in $\Delta_1$ with linear types cannot occur within $\Pi_i$, for $i\neq j$. Let $\St_1,...,\St_n$ be the stores given by the induction hypothesis. Take $\St_1,...,(\St_j,\x=\va),...,\St_n$. That $\vdash \St_j,\x=\va : \Pi_j$ follows from $\Pi_j=\Pi_j^0,\x:\ql\ \Pt$ and: 

\smallskip 

\reg{\vdash \St_j : \Delta_1\circ\Pi_j^0\quad \Delta_1\vdash \va:\ql\ \Pt}{\vdash \St_j,\x=\va :\Pi_j^0,\x:\ql\ \Pt}

\smallskip

\noindent (2) Follows from (1): if $\vdash \St:\Pi_1\sqcup...\sqcup\Pi_n$ then $\vdash \St^\q:\Pi_1^\q\circ ...\circ\Pi^\q_n$, where $\Pi_i^\q$ (resp. $\St^\q$) is obtained from $\Pi_i$ by eliminating variables of type $\qh\ \B$.
\end{proof}

\begin{lem}[Store]
If  $\vdash \St:\Pi,\x:\ql\ \Pt,\Pi_1$, then there exists a context $\Pi^*$  complementary to $\Pi$ such that $\vdash \St\sim_\ql\x:\Pi^*\circ\Pi,\Pi_1$.
\end{lem}

\begin{proof} The proof is by structural induction the derivation of $\vdash \St:\Pi,\x:\ql\ \Pt,\Pi_1$. We take $\Pi_1=\emptyset$ as the base case, which follows since $\vdash\St\sim_\ql\x:\Pi^*\circ\Pi$ and $\Pi^*\vdash \va:\ql\ \Pt$ are the premises of the derivation of $\vdash \St:\Pi,\x:\ql\ \Pt$.
For the inductive step, consider the derivation:

 \smallskip

$
\begin{array}{ll}
\underline{\vdash\St,\x=\va,\St_1:\Delta_1\circ(\Pi,\x:\ql\ \Pt,\Pi_1)\quad \Delta_1\vdash \va':\T}\\
\vdash \St,\x=\va,\St_1,\x'=\va':\Pi,\x:\ql\ \Pt,\Pi_1,\x':\T
\end{array}
$.

\smallskip

\noindent Let $\overline{\Pi},\overline{\Pi}_1$ such that $\overline{\Pi},\x:\ql\ \Pt,\overline{\Pi}_1=\Delta_1\circ(\Pi,\x:\ql\ \Pt,\Pi_1)$ and $\vdash\St,\x=\va,\St_1:\overline{\Pi},\x:\ql\ \Pt,\overline{\Pi}_1$. By the induction hypothesis, there exists $\Pi^*$ complementary to  $\overline{\Pi}$ such that $\vdash\St,\St_1:\Pi^*\circ\overline{\Pi},\overline{\Pi}_1$. Thus $\Pi^*$ and $\Delta_1$ do not have linear variables in common. Then $\Pi^*\circ\overline{\Pi},\overline{\Pi}_1 = \Delta_1\circ(\Pi^*\circ\Pi,\Pi_1)$ and $\vdash\St,\St_1,\x'=\va':\Pi^*\circ\Pi,\Pi_1,\x':\T$ follows.
\end{proof}

\begin{cor}[Store]
If  $\vdash \St:\Pi,\x:\ql\ \B,\Pi_1$  then $\vdash \St\sim_\ql\x:\Pi,\Pi_1.$ \qed\end{cor}

\noindent Lemma Context follows easily by structural induction. It will be key to use Lemma Split.

\begin{lem}[Context] If  $\vdash\St:\Pi$ and $\Pi\vdash\e[\e_0]:\T$, then there exist $\St_0,\Pi_0,\T_0$ such that  $\vdash \St_0:\Pi_0$, $\St_0$ is a subsequence of $\St$, and  $\Pi_0\vdash\e_0:\T_0$. \qed \end{lem}

\begin{lem}[Sustitution] If
\begin{enumerate}
 \item $\Delta_1\vdash \ql\ \La\x_1:\T_1,...,\x_n:\T_n\Ra$
 \item $\Delta_2,\z_1:\T_1,...,\z_n:\T_n\vdash \e:\T$
\end{enumerate}

\noindent then $\Delta_1\circ \Delta_2\vdash [\z_1\mapsto\x_1,...,\z_n\mapsto\x_n]\e:\T$.
\end{lem}

\begin{proof} Note that since $\Delta_1\circ \Delta_2$ must be well defined, then for every $i=1,...,n$ we have that either $\x_i\notin dom\ \!\Delta_2$ or $\T_i$ is unrestricted and  $\Delta_2\ \!\x_i=\T_i$. The proof is by structural induction on $\e$.  The delicate  part of the argument deals with the binding constructors. Consider the case $\Spl\ \e\ \As\ \La\w_1,...,\w_k\Ra\ \In$ $\e'$.  Define the contexts $\Delta_{21},\Delta_{22}$ and the sequences $\z_{i_1},\z_{i_2}$ (subsequences of $\z_1,...,\z_n$)  such that (a) $\Delta_2= \Delta_{21}\circ\Delta_{22}$,

\smallskip
(b) $\z_1:\T_1,...,\z_n:\T_n= (\z_{i_1}:\T_{i_1})_{i_1}\circ(\z_{i_2}:\T_{i_2})_{i_2}$
\smallskip

\noindent and:

\smallskip
\noindent$
\begin{array}{c}
\underline{ \Delta_{21},(\z_{i_1}:\T_{i_1})_{i_1}\vdash \e:\q\ \La\T^\w_1,...,\T^\w_k\Ra
\quad\Delta_{22},(\z_{i_2}:\T_{i_2})_{i_2},\w_1:\T^\w_1,...,\w_k:\T^\w_k\vdash \e':\T}\\
\Delta_2,\z_1:\T_1,...,\z_n:\T_n\vdash
\quad\Spl\  \e\ \As\ \La\w_1,...,\w_k\Ra\ \In\ \e':\T\\
\end{array}
$

\smallskip

\noindent  The last rule is correct because since $\T_i$ is a type (not a pseudotype), for all $i=1,...,n$, then in (a) and (b) one can replace $\circ$ by $\sqcup$. Let $\w'_j$ be a rename of $\w_j$ to avoid capture, for $j=1,..,k$. 
Define $\Delta^\x_1=\Delta^\qu_{21},(\x_{i_1}\!:\!\T_{i_1})_{i_1}$ and $\Delta^\x_2=\Delta^\qu_{22},(\x_{i_2}\!:\!\T_{i_2})_{i_2}$. 
It is easy to check that:

\medskip

$\Delta^\x_2,\w'_1\!:\!\T^\w_1,...,\w'_k\!:\!\T^\w_k\vdash\ql\ \La...,\x_{i_2},...,\w'_1,...,\w'_k\Ra:\ql\ \La...,\T_{i_2},...,\T^\w_1,...,\T^\w_k\Ra$

\smallskip

\noindent  By the induction hypothesis, we have that:
\smallskip

(A) $\Delta^\x_1\circ\Delta_{21}\vdash [\z_{i_1}\mapsto\x_{i_1}]\e:\q\ \La\T^\w_1,...,\T^\w_k\Ra$

\smallskip

(B) $(\Delta^\x_2,\w'_1\!:\!\T^\w_1,...,\w'_k\!:\!\T^\w_k)\circ\Delta_{22}\vdash [\z_{i_2}\mapsto\x_{i_2},\w_j\mapsto\w'_j]\e':\T$

\smallskip

\noindent  Since $\w'_j$'s are fresh variables, we have that 
\smallskip

$(\Delta^\x_2,\w'_1\!:\!\T^\w_1,...,\w'_k\!:\!\T^\w_k)\circ\Delta_{22}=(\Delta^\x_2\circ\Delta_{22}),\w'_1\!:\!\T^\w_1,...,\w'_k\!:\!\T^\w_k$. 

\smallskip

\noindent It is easy to check that $\Delta_1\circ\Delta_{2}=(\Delta^\x_1\circ\Delta_{21})\circ(\Delta^\x_2\circ\Delta_{22})$. On the other hand we have:

\smallskip

$[\z_j\mapsto\x_j](\Spl\  \e\ \As\ \La\w_1,...,\w_n\Ra\ \In\ \e')\Spl\  [\z_j\mapsto\x_j]\e\ \As\ \La\w'_1,...,\w'_n\Ra\ \In\ [\z_j\mapsto\x_j,\w_j\mapsto\w'_j]\e'$.

\smallskip

\noindent Finally, the case follows from the rule for split: 

\smallskip

\regl{\Delta^\x_1\circ\Delta_{21}\vdash [\z_j\mapsto\x_j]\e:\q\ \La \T^\w_1,...,\T^\w_k\Ra\smallskip} 
{(\Delta^\x_2\circ\Delta_{22}),\w'_1\!:\!\T^\w_1,...,\w'_k\!:\!\T^\w_k\vdash [\z_j\mapsto\x_j,\w_j\mapsto\w'_j]\e'\smallskip}{\Delta_{1}\circ\Delta_{2}\vdash[\z_j\mapsto\x_j](\Spl\  \e\ \As\ \La\w_1,...,\w_k\Ra\ \In\ \e'):\T\qquad\ }
\end{proof}

\begin{thm}[Preservation] If $\vdash(\St,\e)$ and $(\St,\e)\rightarrow(\St',\e')$,  then $\vdash(\St',\e')$.
\end{thm}

\begin{proof} Let $\e=\e[\e_0],\ \e'=\e[\e'_0] $, and suppose that $\vdash \St:\Pi$ and $\Pi\vdash\e[\e_0]:\T$. The proof is by induction on the context $\e[]$. We will prove that there exist $\St',\Pi'$ such that $\vdash \St':\Pi'$ and $\Pi'\vdash\e[\e'_0]:\T$.

\textbf{Case} $\e[] = []$. We will prove the case $\e_0=\Spl\ \x\ \As\ \La\z_1,...,\z_n\Ra\ \In\ \e_1$. Let $\St\x=\q\ \La\x_1,...,\x_n\Ra$. We have that that:

$(\St, \Spl\ \x\ \As\ \La\z_1,...,\z_n\Ra\ \In\  \e_1) \rightarrow_\beta (\St\sim_\q\x, [\La\z_1,...,\z_n\Ra\mapsto\La\x_1,...,\x_n\Ra]\e_1)$

\noindent Let $\e'=[\La\z_1,...,\z_n\Ra\mapsto\La\x_1,...,\x_n\Ra]\e_1$ and $\St'=\St\sim_\q \x$. We know that $\Pi=\Pi_1\sqcup\Pi_2$ and:

\smallskip \reg{\Pi_1\vdash \x:\q\ \La\T_1,...,\T_n\Ra\quad \Pi_2,\z_1:\T_1,...,\z_n:\T_n\vdash \e_1:\T}
{\Pi_1\sqcup\Pi_2\vdash \Spl\ \x\ \As\ \La\z_1,...,\z_n\Ra\ \In\  \e_1:\T}

\smallskip

If $\q=\qu$, then $\qu(\T_i)$ for all $i$.
If we apply the Substitution Lemma by defining $\Delta_1=\Pi_\qu$ and $\Delta_2=\Pi$, then we obtain: $\Pi \vdash[\La\z_1,...,\z_n\Ra\mapsto\La\x_1,...,\x_n\Ra]\e_1:\T $. Taking $\Pi'=\Pi$, the case follows immediately.

If $\q=\ql$, let $\Pi = \Delta_0,\x:\ql\ \La\T_1,...,\T_n\Ra,\Delta$. By Lemma Store  there exists $\Delta_1$ complementary to $\Delta_0$ such that $\Delta_1\vdash \ql\ \La\x_1,...,\x_n\Ra : \ql\ \La\T_1,...,\T_n\Ra $ and $\vdash\St\sim_\ql \x:\Delta_0\circ\Delta_1,\Delta$. Then $\Pi_1 = \Delta^\qu_0,\x:\La\T_1,...,\T_n\Ra,\Delta^\qu$ and $\Pi_2 =\Delta_0,\Delta$. Since  $\Delta_1,\Delta^\qu\vdash \ql\ \La\x_1,...,\x_n\Ra : \ql\ \La\T_1,...,\T_n\Ra $, by Lemma Sustitution, $(\Delta_1,\Delta^\qu)\circ(\Delta_0,\Delta)\vdash [\La\z_1,...,\z_n\Ra\mapsto\La\x_1,...,\x_n\Ra]\e_1:\T$. The case follows taken  $\Pi' = \Delta_0\circ\Delta_1,\Delta=(\Delta_1,\Delta^\qu)\circ(\Delta_0,\Delta)$.

\smallskip

\textbf{Case} $\e[] = \q\ \La\x_1,...,\e_k[],...,\e_n\Ra$. If $\St=\St'$ then the case follows easily.

Suppose that $\St'=\St\sim_\ql \x$ and $\Pi = \Delta_0,\x:\ql\ \Pt,\Delta$, where $\Pt$ is not a base type (cases split and application). First we prove that there exists $\Pi'$ such that $\vdash\St':\Pi'$.
By Lemma Store there exists $\Delta_1$ complementary to $\Delta_0$ such that $\vdash\St':\Delta_0\circ\Delta_1,\Delta$. Then, take $\Pi'=\Delta_0\circ\Delta_1,\Delta$ to obtain $\vdash\St':\Pi'$.

Finally, we prove for the case $\St'=\St\sim_\ql \x$ that if  $\Pi\vdash\e[\e_0]:\T$ then $\Pi'\vdash\e[\e'_0]:\T$. Since $\Pi=\Pi_1\sqcup...\sqcup\Pi_n$, by   Lemma Split there exist $\St_1,...,\St_n$ such that  $\vdash \St_j:\Pi_j$, for all $j=1,...,n$. By the  induction hypothesis  there exist  $\Pi'_k,\St'_k$ such that $\Pi'_k\vdash \e_k[\e'_0]:\T_k$ and $\vdash \St'_k:\Pi_k'$. It remains to prove that $\Pi'=\Pi_1\sqcup...\sqcup\Pi'_k\sqcup...\sqcup\Pi_n$. This fact follows from the fact that  $\St_1,...,\St_n$ are subsequences of $\St$ that, taken in pairs, do not have variables with linear values in common, except those of base type.

Suppose now that $\St'=\St\sim_{\sq_1,...,\sq_m} \x_1,...,\x_m,\x=\q'\ b$ and $\Pi = \Delta_0,(\x_j:\sq_j\ \B_j)_j,\Delta$ (case basic operator). Without loss of generality (for a simpler notation) we can assume that $\sq_j=\ql$ for all $j$.
By Corollary Store (applied $m$ times) $\vdash \St\sim_{\sq_1,...,\sq_m} \x_1,...,\x_m:\Delta_0,\Delta$. Then, take $\Pi'=\Delta_0,\Delta, \x:\q'\ \B$ to obtain $\vdash\St':\Pi'$.

It remains to check that $\Pi'\vdash\e[\e'_0]:\T$. It is trivial  except in the case $\q'=\ql$. Since $\Pi=\Pi_1\sqcup...\sqcup\Pi_n$, by Lemma Split there exist $\St_1,...,\St_n$ such that  $\vdash \St_j:\Pi_j$, for all $j=1,...,n$. By the  induction hypothesis  there exist  $\Pi'_k,\St'_k$ such that $\Pi'_k\vdash \e_k[\e'_0]:\T_k$ and $\vdash \St'_k:\Pi_k'$. To conclude the case we must to take $\St'_j=(\St_j,\x=\ql \ b)$ and $\Pi_j=(\Pi_j,\x:\qh\ \B)$, for $j=1,,,.k-1$.  Thus we can prove that  $\Pi'=(\Pi_1,\x:\qh\ \B)\sqcup...\sqcup (\Pi_{k-1},\x:\qh\ \B)\sqcup\Pi'_k\sqcup...\sqcup\Pi_n$. 
\end{proof}

\begin{cor}[Progress] If $\vdash(\St,\e)$ and $\e$ is not a  variable, then there exist  $\St',\e'$ such that  $(\St,\e)\rightarrow(\St',\e')$.
\end{cor}

\begin{proof} Follows lemma Context and corollary $\beta$-progress. 
\end{proof}

\section{Performance of weak-linearity}
\label{Rendimientodelalinealidaddebil}

We are going to show some examples of well-typed weak-linear programs that will allow us to quantify the efficiency in the use of memory resources that weak-linearity provides.

For a better readability of the program, we remove the qualifiers from each operation $\op^\tau$ (we just write $\op$). In the tables $fib$, $mapa$, $map$ and $sort$ we write on the left $\St$, $\e$ and $\Pi$, and on the right the signature $\Sigma^\q$ with its operators in order of appearance. In all cases $\vdash(\St,\e)$ is verified.

To test weak-linearity in non-trivial programs, we are going to allow recursive functions by adding self-referring definitions in the store $\St$ (they need to be unrestricted). It is not necessary to change the small-step semantics, but it is necessary to add the following rule for store typing:
\[
\begin{array}{c}
\underline{\vdash \St: \Pi\quad \Pi^\qu,\f:\qu\ (\T\rightarrow\T')\vdash \va:\qu\ (\T\rightarrow\T')}\\
\vdash (\St,\f=\va):(\Pi,\f:\qu\ (\T\rightarrow\T'))\\
\end{array}
\]

On the other hand, we can add recursive types using a similar approach to tuples. We add lists to our language using the syntax shown in Figure 9.

\smallskip

\begin{center}
\begin{tabular}{|c |}
\hline
$
\begin{array}{llllll}
&\Pt \ ::= \ \q\ \List\ \T & \mathsf{list\ pretype}\smallskip\\
&\w \ \ ::=  \ []\ |\ (\x:\x)&\mathsf{list\ prevalues}\smallskip\\
&\e \ \ ::=  \ \q\ \![]\ |\ \q\ (\e:\e)\ |\  \Case\ \e\ \Of\ (\e,(\x\!:\!\x)\!\rightarrow\!\e)\quad&\mathsf{list\ expressions}\smallskip\\
&\e[]\ ::=\ \q\ (\e[]\!:\!\e)\ |\ \q\ (\x\!:\!\e[])\ | \  
\Case\ \e[]\ \Of\ (\e,(\x:\x)\!\rightarrow\!\e) & \mathsf{list\ context}\quad\bigskip\\
\mathsf{(eca)} & (\St,\Case\ \x\ \Of\ (\e_0,(\z_1,\z_2)\rightarrow\e_1) \rightarrow_\beta\smallskip \\
& \qquad\qquad\qquad\qquad\qquad\qquad (\St\sim_\q\x,\e_0)& (\St\x=\q\ [])\\
  & \qquad\qquad\qquad\qquad\qquad\qquad(\St\sim_\q\x,[\langle\z_1,\z_2\rangle\mapsto\langle\x_1,\x_2\rangle]\e_1)&(\St\x=\q\ (\x_1:\x_2))\smallskip\\

\end{array}
$\\
\hline
\end{tabular}
\tiny Figura 9: Syntax and semantics of lists
\end{center}

\medskip

We will denote by $C^\ql_{\f} (n)$ the balance\footnote{Number of memory locations created minus the number of memory locations deleted.} of memory locations used by the function $\f$, where $n$ is the number of recursive calls. Each memory location that holds an integer, boolean, list element or tuple counts 1, while an array data counts according to its size. We will use $C^\qu_{\f} (n)$ for the cost of the unrestricted version of $\f$ (all its qualifiers are $\qu$).

In the first program $fib$, the identity function $\id$ is used with the sole objective of allowing a hidden use of the object $\w$, thus ensuring that the program returns two weak-linear data. On the other hand, the variable $\x$ must be returned to guarantee linearity ($\x$ must be used).
In the same way, the operator $\pi_1$ in $map$ allows the effective use of $\x$ and $\n$ to proceed with its deallocation (if we only put $\ql\ []$, the program would not be well typed).

The funtions $\Fib$, $\Map$ and $\Mapa$ are examples of algorithms in which weak-linearity allows optimal efficiency. For these programs we have that $C^\ql (n)$ is $\mathcal{O}(1)$, while $C^\qu (n)$ is $\mathcal{O}(n)$, for $\Fib$ and $\Map$, and $\mathcal{O}(n^2)$ for $\Mapa$.

$\Sort$ shows a significant efficiency in the weak-linear version, and simultaneously shows the impossibility of obtaining a weak-linear version $\mathcal{O}(1)$, due to the strongly nonlinear use of some variables of type integer.
We observe that $C^\ql_{\mathit{\Sort}} (n)$ is $\mathcal{O}(n)$, while $C^\qu_{\mathit{\Sort}} (n)$ is $\mathcal{O}(n^3)$.

\medskip

\noindent\begin{center}
$fib$

\begin{tabular}{|c |c |}
\hline
$
\tiny\begin{array}{llll}
\St &\Fib\  =\ \lambda\x:\ql\ \Int.\\
&\qquad\quad\quad \If\ \x=0\\
&\qquad\quad\quad \Then\ \ql\  \La \x,1,1\Ra\\
&\qquad\quad\quad \Else\ \Spl\ \Fib(\x-1)\  \As \ \La\x,\w,\y\Ra\ \In\ \ql\ \La\x, \id\ \y,\w+\y\Ra  \\
\e &\Fib\ n\smallskip\\
\Pi& \Fib : \ql\ \Int\rightarrow\La\ql\ \Int,\ql\ \Int,\ql\ \Int\Ra\\
\end{array}
$
&
$
\tiny\begin{array}{lcll}
\Sigma^\q&
(=)& : & (\qh\ \Int,\ql\ \Int)\rightarrow\ql\ \Bool\\
&0  & : & \ql\ \Int\\
&1  & : & \ql\ \Int\\
&1  & : & \ql\ \Int\\
&(-)& : & (\ql\ \Int,\ql\ \Int)\rightarrow\ql\ \Int\\
&\id& : & \qh\ \Int\rightarrow\ql\ \Int\\
&(+)& : & \La\ql\ \Int,\ql\ \Int\Ra\rightarrow\ql\ \Int\qquad\quad\\
& n & : & \ql\ \Int
\end{array}
$\\
\hline
\end{tabular}\\
\end{center}

\medskip

\noindent\begin{center}
$mapa$

\begin{tabular}{|c |c |}
\hline
$
\tiny\begin{array}{llll}
\St&\ar = \ql\  \left\lbrace 0,1,2,3,...,n-1\right\rbrace,\ \\
& \Mapa = \qu\ (\lambda\ \f:\qu\ (\ql\ \Int\rightarrow\ql\ \Int).\\
&\qquad\quad \ql\ (\lambda\w:\ql\ \La\ql\ \Array,\ql\ \Int,\ql\ \Int\Ra.\\
&\qquad\qquad\Spl\ \w\ \As \ \La\ar,\vi,\n\Ra\ \In\\
&\qquad\qquad\qquad\If\ \vi=\n\ \Then\ \ql\ \La\ar,\vi,\n\Ra\\
&\qquad\qquad\qquad\Else\ \Let\ \z:\ql\ \Int \equiv\ar[\vi]\\
&\qquad\qquad\qquad\quad\In\ \Mapa (\ql\ \La\ar[\vi\leftarrow\f\ \z],\vi+1,\n\Ra))\\
\e & \Mapa\ ( \qu\ (\lambda\x:\ql\ \Int.\ \x+1))\ (\ql\ \La\ar,0,n\Ra)\smallskip\\
\Pi & \ar:\ql\ \Array,\ \Mapa:\\
& \qu\ (\ql\ \La\ql\ \Array,\ql\ \Int,\ql\ \Int\Ra\rightarrow\ql\ \La\ql\ \Array,\ql\ \Int,\ql\ \Int\Ra)
\end{array}
$
&
$
\tiny\begin{array}{lcll}
\Sigma^\q
& (=)& : & \La\qh\ \Int,\qh\ \Int\Ra\rightarrow\ql\ \Bool\\
& \cdot[\ \cdot\ ] & : & \La\qh\ \Array, \qh\ \Int\Ra\rightarrow\ql\ \Int\\
& \cdot[\cdot\leftarrow\cdot ] & : & \La\ql\ \Array, \qh\ \Int, \ql\ \Int\Ra\rightarrow\ql\ \Array\\
& (+) & : & \La\ql\ \Int,\ql\ \Int\Ra\rightarrow\ql\ \Int\\
 & 1 & : & \ql\ \Int\\
 &(+)& : & \La\ql\ \Int,\ql\ \Int\Ra\rightarrow\ql\ \Int\\
 & 1 & : & \ql\ \Int\\
 & 0 & : & \ql\ \Int\\
 & n & : & \ql\ \Int\\
\end{array}
$\\
\hline
\end{tabular}\\
\end{center}

\medskip

\noindent\begin{center}
$map$

\begin{tabular}{|c |c |}
\hline
$
\tiny\begin{array}{llll}
\St&\List\ =\ \qu\ (\lambda \p: \ql\ \La\ql\ \Int,\ql\ \Int\Ra.\\
&\qquad\quad  \Spl\ \p\ \As\ \La\n,\x\Ra\ \In\\
&\qquad \qquad \If\ n==0 \ \Then\ \pi_1(\pi_1(\ql\ \![],\n),\x)\\
& \qquad\qquad\Else\ \ql\ (\id(\x):\List\ \La \n-1,\x+1\Ra),\ \\
& \Map = \qu\ (\lambda\ \f:\qu\ (\ql\ \Int\rightarrow\ql\ \Int).\\
&\qquad\quad \ql\ (\lambda\xs:\ql\ \List\ (\ql\ \Int).\\
&\qquad\qquad\Case\ \xs\ \Of \ (\ql\ \![],(\z\!:\!\zs)\!\rightarrow\!\ql\ (\f\ \z:\Map\ \f\ \zs) \\
\e & \Map\ ( \qu\ (\lambda\x:\ql\ \Int.\ \x+1))\ (\ql\ \La n,0\Ra)\smallskip\\
\Pi & \List:\qu\ (\ql\ \La\ql\ \Int,\ql\ \Int\Ra\rightarrow\ql\ \List\ (\ql\ \Int)),\ \Map:\\
& \qu\ (\qu\ (\ql\ \Int\rightarrow\ql\ \Int)\rightarrow\ql\ \List\ (\ql\ \Int)\rightarrow\ql\ \List\ (\ql\ \Int))
\end{array}
$
&
$
\tiny\begin{array}{llcll}
\Sigma^\q
&== & : & \qh\ \Int\rightarrow\ql\ \Bool\\
&\pi_1 & : & (\ql\ [\ql\ \Int],\ql\ \Int)\rightarrow\ql\ [\ql\ \Int]\qquad\\
&\pi_1 & : & (\ql\ [\ql\ \Int],\ql\ \Int)\rightarrow\ql\ [\ql\ \Int]\\
&\id & : & \qh\ \Int\rightarrow\qu\ \Int\\
&- & : & (\ql\ \Int,\ql\ \Int)\rightarrow\ql\ \Int\\
&1 & : & \ql\ \Int\\
&+ & : & (\ql\ \Int,\ql\ \Int)\rightarrow\ql\ \Int\\
&1 & : & \ql\ \Int\\
&+ & : & (\ql\ \Int,\ql\ \Int)\rightarrow\ql\ \Int\\
&1 & : & \ql\ \Int\\
&n & : & \ql\ \Int\\
&0 & : & \ql\ \Int\smallskip\\
\end{array}
$\\
\hline
\end{tabular}\\
\end{center}

\medskip


\noindent\begin{center}
$sort$

\begin{tabular}{|c |c |}
\hline
$
\tiny\begin{array}{llll}
\St & \ar= \ql\  \left\lbrace n-1,...,1,0\right\rbrace,\\
& \Swap =\qu\ (\lambda\w:\ql\ \La\ql\ \Array,\qu\ \Int,\qu\ \Int\Ra.\\
& \qquad\qquad\quad\Spl\ \w\ \As\
\La\ar,\vi,\vj\Ra\ \In\\
& \qquad\qquad\qquad\Let\ \z\equiv\ar[\vi]\ \In\ \Let\ \y\equiv\ar[\vj]\ \In\\
& \qquad\qquad\qquad \Let\ \ar\equiv\ar[\vi\leftarrow\y]\ \In\ \Let\ \ar\equiv\ar[\vj\leftarrow\w])\\
& \Ins =\qu\ (\lambda\z.\ \Spl\ \z\ \As \ \La\ar,\vi\Ra\ \In\\
&\qquad\qquad\qquad\If\ \vi=0\ \Then\ \ar\\
&\qquad\qquad\qquad\Else\ \Let\ \vj \equiv\vi-1\ \In\\
&\qquad\qquad\qquad\qquad\If\ \ar[\vi]<\ar[\vj] \\
&\qquad\qquad\qquad\qquad\Then\ \Ins (\ql\ \La\Swap\ (\ql\ \La\ar,\vi,\vj\Ra),\vj\Ra)\\
&\qquad\qquad\qquad\qquad\Else \
\Ins (\ql\ \La\ar,\vj\Ra),\\
& \Sort= \qu\ (\lambda\w.\ \Spl\ \w\ \As \ \La\ar,\vi,\vj\Ra\ \In\\
&\qquad\qquad\qquad\If\ \vi=\vj\ \Then\ \ar\\
&\qquad\qquad\qquad\Else\
\Sort (\ql\ \La\Ins\ (\ql\ \La\ar,\vi\Ra),\vi+1,\vj\Ra)\\
\e & \Sort (\ql\ \La\ar,0,n\Ra)\\
\Pi & \ar:\ql\ \Array,\ \Swap : \ql\ \La\ql\ \Array,\qu\ \Int,\qu\ \Int\Ra\rightarrow\ql\ \Array,\\
& \Ins   : \La\ql\ \Array,\qu\ \Int\Ra\rightarrow\ql\ \Array,\\ &\Sort  : \La\ql\ \Array,\qu\ \Int,\qu\ \Int\Ra\rightarrow\ql\ \Array
\end{array}
$
&
$
\tiny\begin{array}{lrll}
\Sigma^\q
& \cdot[\ \cdot\ ]  & : & \La\qh\ \Array, \qu\ \Int\Ra\rightarrow\ql\ \Int\\
& \cdot[\ \cdot\ ]  & : & \La\qh\ \Array, \qu\ \Int\Ra\rightarrow\ql\ \Int\\
&  \cdot[\cdot\leftarrow\cdot ] & : & \La\ql\ \Array, \qu\ \Int, \ql\ \Int\Ra\rightarrow\ql\ \Array\\
&  \cdot[\cdot\leftarrow\cdot ]& : & \La\ql\ \Array, \qu\ \Int, \ql\ \Int\Ra\rightarrow\ql\ \Array\\
& (=)& : & \La\qu\ \Int,\ql\ \Int\Ra\rightarrow\ql\ \Bool\\
& 0 & : & \ql\ \Int\\
& (-)& : & \La\qu\ \Int,\ql\ \Int\Ra\rightarrow\qu\ \Int\\
& 1 & : & \ql\ \Int\\
& \cdot[\ \cdot\ ]  & : & \La\qh\ \Array, \qu\ \Int\Ra\rightarrow\ql\ \Int\\
& (<)& : & \La\ql\ \Int,\ql\ \Int\Ra\rightarrow\ql\ \Bool\\
& \cdot[\ \cdot\ ]  & : & \La\qh\ \Array, \qu\ \Int\Ra\rightarrow\ql\ \Int\\
& (=)& : & \La\qu\ \Int,\ql\ \Int\Ra\rightarrow\ql\ \Bool\\
& (+)& : & \La\qu\ \Int,\ql\ \Int\Ra\rightarrow\qu\ \Int\\
& 1 & : & \ql\ \Int\\
& 0& : & \qu\ \Int\\
& n& : & \qu\ \Int\\
\end{array}
$\\
\hline
\end{tabular}\\
\end{center}

\section{Conclusions}
\label{conclusiones}


In this work, we present a language that supports a weakened form of linearity, while preserving the simplicity and elegance of the original linear system. We extend the language presented in \cite{walker} by introducing a qualified signature. The only modification we introduce to this linear system is to add the qualifier $\qh$, which allows read-only access to a base linear data, and introduce the pseudo-split of contexts. We only use the traditional split for basic operators. The key is that read-only access to a base linear data is only performed when it appears as a hidden input of a basic operation of the qualified signature. Case studies are shown in which we can observe a significant improvement in the use of memory resources. 

We consider the contribution of weak-linear types valuable because it provides theoretical clarity in the following sense: a type system is obtained that weakens the notion of linearity and simultaneously preserves in its definition the original idea of (only) weakening structural principles through the introduction of context splitting.


We consider that the work proposes a theoretical framework that expresses in a simple way some benefits of linear systems, and that it can be an adequate framework to study the relationship between substructurality and in-place update, a relationship that, although it has been addressed in numerous works, still presents clear challenges.


\bibliographystyle{alphaurl}
\bibliography{wlt}

\begin{thebibliography}{AFKT03}

\bibitem[AFKT03]{aiken}
A.~Aiken, J.~Foster, J.~Kodumal, and T.~Terauchi.
\newblock {\em Checking and inferring local non-aliasing. In ACM SIGPLAN Conference on Programming Language Design and Implementation (PLDI)}.
\newblock San Diego, California, pages 129-140, 2003.

\bibitem[AH02]{aspinall}
D.~Aspinall and M.~Hofmann.
\newblock {\em Another Type System for In-Place Update}.
\newblock D. Le Metayer (Ed.): ESOP 2002, LNCS 2305, pp. 36-52, 2002.

\bibitem[CP02]{cervesato}
I.~Cervesato and F.~Pfenning.
\newblock {\em A linear logical framework}.
\newblock Information and Computation, 179(1):19-75, 2002.

\bibitem[Dos02]{dosen}
K.~Dosen.
\newblock {\em A historical introduction to substructural logics}.
\newblock In ACM SIGPLAN Conference on Programming Language Design and Implementation (PLDI), Berlin, Germany, pages 1-12, 2002.

\bibitem[FTA02]{foster}
J.~Foster, T.~Terauchi, , and A.~Aiken.
\newblock {\em Flow-sensitive type qualifiers}.
\newblock In ACM SIGPLAN Conference on Programming Language Design and Implementation (PLDI), Berlin, Germany, pages 1-12, 2002.

\bibitem[Gir87]{girard}
J~Girard.
\newblock {\em Linear logic}.
\newblock Theoretical Computer Science 50 1-102, 1987.

\bibitem[Kob99]{kobayashi}
N.~Kobayashi.
\newblock {\em Quasi-Linear Types}.
\newblock In Proceedings ACM Principles of Programming Languages, pages 29-42, 1999.

\bibitem[Ode92]{odersky}
M.~Odersky.
\newblock {\em Observers for linear types}.
\newblock In B. Krieg-Bruckner, editor, ESOP ?92: 4th European Symposium on Programming, Rennes, France, Proceedings, pages 390-407. Springer-Verlag, February 1992. Lecture Notes in Computer Science 582, 1992.

\bibitem[SWM00]{smith}
F.~Smith, D.~Walker, and G.~Morrisett.
\newblock {\em Alias types}.
\newblock In European Sym-posium on Programming (ESOP), Berlin, Germany, volume 1782 of Lecture Notes in Computer Science, pages 366-381. Springer-Verlag, 2000.

\bibitem[Wad90]{wadler}
P.~Wadler.
\newblock {\em Linear types can change the world!}
\newblock In IFIP TC 2 Working Conference on Programming Concepts and Methods, Sea of Galilee, Israel, 1990.

\bibitem[Wal05]{walker}
D.~Walker.
\newblock {\em Subestructural Types Systems}.
\newblock in Advanced Topics in Types and Programming Languages. Benjamin C. Pierce, editor. The MIT Press, Cambridge, Massachusetts London, England, 2005.

\end{thebibliography}
\end{document}